\RequirePackage[l2tabu,orthodox]{nag}
\documentclass
[11pt,letterpaper]
{article} 

\usepackage[notes=true,later=false,camera=false]{dtrt}
\usepackage[utf8]{inputenc}
\usepackage{etex}
\usepackage{thmtools}
\usepackage{ stmaryrd }
\usepackage{xspace,enumerate}
\usepackage[T1]{fontenc}
\usepackage[full]{textcomp}
\usepackage[american]{babel}
\usepackage{mathtools}
\usepackage{amsthm}

\usepackage{empheq}

\hypersetup{
colorlinks=true,
urlcolor=Cerulean,
linkcolor=RoyalBlue,
citecolor=OliveGreen,
linktocpage=true,
}
\renewcommand*\backref[1]{\ifx#1\relax \else (pg. #1) \fi}

\usepackage[capitalise,nameinlink]{cleveref}
\crefname{lemma}{Lemma}{Lemmas}
\crefname{fact}{Fact}{Facts}

\crefname{theorem}{Theorem}{Theorems}
\crefname{corollary}{Corollary}{Corollaries}
\crefname{claim}{Claim}{Claims}
\crefname{example}{Example}{Examples}
\crefname{algorithm}{Algorithm}{Algorithms}
\crefname{problem}{Problem}{Problems}
\crefname{definition}{Definition}{Definitions}
\crefname{equation}{Eq.}{Eq.}
\crefname{strategy}{Strategy}{Strategies}

\usepackage{paralist}
\usepackage{turnstile}
\usepackage[framemethod=TikZ]{mdframed}
\mdfsetup{frametitlealignment=\center}
\usepackage{tikz}
\usepackage{caption}
\DeclareCaptionType{Algorithm}
\usepackage{newfloat}
\usepackage{aliascnt}
\newtheorem{theorem}{Theorem}[section]
\newaliascnt{lemma}{theorem}
\newtheorem{lemma}[lemma]{Lemma}
\aliascntresetthe{lemma}\newtheorem*{lemma*}{Lemma}
\newaliascnt{remark}{theorem}
\newtheorem{remark}[remark]{Remark}
\aliascntresetthe{remark}\newtheorem*{remark*}{Remark}
\newaliascnt{definition}{theorem}
\newtheorem{definition}[definition]{Definition}
\aliascntresetthe{definition}\newtheorem*{definition*}{Definition}

\newaliascnt{proposition}{theorem}

\aliascntresetthe{proposition}
\newaliascnt{fact}{theorem}
\newtheorem{fact}[fact]{Fact}
\aliascntresetthe{fact}
\newaliascnt{corollary}{theorem}

\aliascntresetthe{corollary}

\newtheorem{algorithm-thm}[theorem]{Algorithm}

\usepackage[
letterpaper,
top=1.2in,
bottom=1.2in,
left=1in,
right=1in]{geometry}
\usepackage{newpxtext} % T1, lining figures in math, osf in text
\usepackage{textcomp} % required for special glyphs
\usepackage{mathpazo}
\usepackage[scr=rsfso]{mathalfa}% \mathscr is fancier than \mathcal
\usepackage{bm} % load after all math to give access to bold math
\let\mathbb\varmathbb
\usepackage{microtype}

\usepackage{footnotebackref}

\allowdisplaybreaks
\newcommand{\FormatAuthor}[3]{
\begin{tabular}{c}
#1 \\ {\small\texttt{#2}} \\ {\small #3}
\end{tabular}
}

\newcommand{\R}{{\mathbb R}}

\newcommand{\specrad}{\rho}
\newcommand{\eps}{\varepsilon}

\newcommand{\E}{{\mathbb E}}

\newcommand{\1}{\mathbf{1}}
\newcommand{\ip}[1]{\langle #1 \rangle}

\newcommand{\tr}{\mathrm{tr}}

\newcommand{\pE}{\tilde{\E}}
\newcommand{\cH}{\mathcal H}

\newcommand{\cB}{\mathcal B}

\newcommand{\cS}{\mathcal S}

\newcommand{\poly}{\mathrm{poly}}

\newcommand{\mper}{\,.}
\newcommand{\mcom}{\,,}

\newcommand{\Norm}[1]{\left\lVert#1\right\rVert}

\newcommand{\cD}{\mathcal D}

\renewcommand{\emptyset}{\varnothing}
\renewcommand{\geq}{\geqslant}
\renewcommand{\leq}{\leqslant}
\renewcommand{\epsilon}{\varepsilon}

\usepackage{macros}

\begin{document}
\title{An Approximate Cauchy-Schwarz Inequality and Improved Bounds for Sherali-Adams Refutation of Semirandom CSPs}

%\author{Anonymous}
\author{
 \begin{tabular}{cc}
  \FormatAuthor{Pravesh K. Kothari}{kothari@cs.princeton.edu}{Princeton University} &
\FormatAuthor{Andrew D. Lin}{andrewlin@princeton.edu}{Princeton University}
  \end{tabular}
  }
\maketitle

\begin{abstract}
We formulate an \emph{approximate Cauchy-Schwarz} inequality and show that it is satisfied by solutions to the Sherali-Adams linear programming hierarchy (interpreted as ``pseudo-distributions''). As a consequence, we resolve a question left open by the work of O'Donnell and Schramm~\cite{OS19} that they had explicitly attributed to the lack of such an inequality. 

A Cauchy-Schwarz inequality is \emph{exactly} satisfied by pseudo-distributions satisfying the constraints of the \emph{sum-of-squares} semidefinite programming hierarchy and already has scores of applications. However, the proof there requires \emph{global} positive semidefiniteness. Our approximate version, on the other hand, relies only on \emph{local} positive semidefiniteness satisfied by the Sherali-Adams pseudo-distributions. Our formulation loses an additive error that scales with the L1 norm of the coefficients of the constituent polynomials, and this loss is asymptotically tight. Our proof is elementary and relies on a simple sampling argument. 

As an application, we resolve a question left open in the work of O'Donnell and Schramm that gives a trade-off between constraint density and the Sherali-Adams degree for refuting random constraint satisfaction problems. Specifically, for odd arity CSPs, we show that the constraint density requirement for a given degree can be improved by a polynomial factor in $n$. Along the way, we observe that by a simple extension, the results in their work extend to a more general semirandom setting.
\end{abstract}

\section{Introduction}
A degree $d$ sum-of-squares (SoS) pseudo-distribution on $\{-1,1\}^n$ is a signed measure $D:\{-1,1\}^n \rightarrow \R$ such that the associated pseudo-expectation functional defined for any $f:\{-1,1\}^n \rightarrow \R$ by $\pE_{D}[f] = \sum_{x} D(x) f(x)$ satisfies 1) $\pE_{D} [ 1] = 1$ and 2) $\pE_{D}[ f^2] \geq 0$ for every polynomial $f$ of degree $\leq d/2$. Such pseudo-distributions are equivalent to feasible points for SDPs from the SoS semidefinite programming hierarchy. Degree $d$ SoS pseudo-distributions are equivalent to \emph{pseudo-moment matrices} -- matrices $M$ indexed by subsets of $[n]$ of size $\leq d/2$ whose entry at $(S,T)$ is given by $\pE_{D}[ \prod_{i \in S \Delta T} x_i]$. Pseudo-moment matrices are themselves in one-to-one correspondence with the set of all positive semidefinite matrices that satisfy some natural linear constraints\footnote{These constraints encode that $M(S,T)$ depends only on $S \Delta T$ and the diagonals of $M$ are $1$.}. The SoS semidefinite programming hierarchy, at a high level, can be thought of as solving (in $n^{O(d)}$ time) the problem of finding a pseudo-expectation (equivalently, a pseudo-moment matrix) satisfying some additional linear constraints. 

The analysis of algorithms (see the survey ~\cite{FlemingKP19}) based on the SoS hierarchy requires reasoning about the pseudo-expectations of various polynomials arising in a given problem. Such analyses often rely on generalizing natural analytic inequalities that hold for expectations (w.r.t. probability distributions) to pseudo-expectations (w.r.t. pseudo-distributions). Indeed, the power of the pseudo-expectation view of SoS solutions comes largely from being able to ``guess" true inequalities for pseudo-distributions by analogy with actual probability distributions.

\paragraph{The Cauchy-Schwarz Inequality} Perhaps the most fundamental of such general analytic inequalities is the Cauchy-Schwarz inequality:

\[
\pE_{D}[ fg] \leq \pE_D[ f^2]^{1/2} \pE_{D}[g^2]^{1/2},
\]
which is satisfied whenever $D$ is a degree $d$ pseudo-distribution and $f,g$ are polynomials of degree $\leq d/2$. The Cauchy-Schwarz inequality appears in scores of applications of the sum-of-squares hierarchy, especially in executing the \emph{proofs-to-algorithms} paradigm in algorithm design. The proof (and indeed, the truth) of the inequality relies on the positive semidefiniteness constraint that is often viewed as a \emph{global} constraint on the pseudo-moments. 

Such global positive semidefiniteness constraints are not satisfied by solutions to the weaker \emph{Sherali-Adams} linear-programming hierarchy. An $R$-local Sherali-Adams solution can also be viewed as a degree $2R$ pseudo-distribution $D: \{-1,1\}^n \rightarrow \R$ with the key difference that the associated pseudo-expectation only satisfies a \emph{local} version of (2). More precisely: 

\begin{definition}[$R$-local pseudo-distributions and pseudo-expectation]
An $R$-local pseudo-distribution over the $n$-dimensional Boolean hypercube is a function $D:\{-1,1\}^n \rightarrow \R$ such that the associated pseudo-expectation defined by $\pE_D[ f] := \sum_{x \in \{-1,1\}^n} D(x) f(x)$ satisfies:
\begin{enumerate}
\item $\pE[ 1] = 1$,
\item For every pointwise non-negative $f: \{-1,1\}^n \rightarrow \R$ such that $f$ depends on only $\leq R$ out of $n$ possible inputs, we have: $\pE[ f] \geq 0$. 
\end{enumerate}
\end{definition}
Equivalently, for every non-negative function $f$ that depends only on $\leq R$ variables, $\pE[ f] \geq 0$. The corresponding pseudo-moment matrices now satisfy \emph{local} positive semidefiniteness constraints -- every principal submatrix that corresponds to pseudo-moments of any set of at most $R$ variables is positive semidefinite. However, the matrix as a whole may not be PSD! This weakening allows encoding the constraints of a degree-$d$ Sherali-Adams pseudo-expectation as an $n^{O(d)}$-size \emph{linear} (as opposed to a semidefinite) program. The weakening to only satisfying local constraints comes at a cost. Results like the Cauchy-Schwarz inequality no longer hold, provably (see \Cref{thm:sa-matching-lower-bound})!

Our goal in this work is to formulate an approximate version of the Cauchy-Schwarz inequality that can hold for Sherali-Adams pseudo-distributions. Before diving into our formulation, let us discuss our main motivation:

\paragraph{Sherali-Adams vs Sum-of-Squares} How does the power of the Sherali-Adams linear programming hierarchy compare with that of the SoS SDP hierarchy? On the one hand, lower bounds for constraint satisfaction problems such as Max-Cut~\cite{CharikarMakarychevs} show that degree $2$ SoS SDPs~\cite{GoemansWilliamson} can achieve drastically better approximation ratios even as Sherali-Adams LPs of $n^{\epsilon}$-degree fail to appreciably beat simply returning a uniformly random assignment. On the other hand, O'Donnell and Schramm~\cite{OS19} and the follow-up work due to Hopkins, Schramm, and Trevisan~\cite{HopkinsST} show that degree $n^{c}$ (for fixed $c<1$) Sherali-Adams LPs can qualitatively match the performance of SDPs for Max-Cut! Perhaps more strikingly, they show that $O(1/\epsilon)$-degree Sherali-Adams LPs (thus, in polynomial time!) can strongly refute (i.e., certify an upper bound of $1/2+\epsilon$ on the normalized value of) random $k$-XOR instances with $m \geq n^{k/2+\epsilon}$ constraints \textbf{for even $k$}. This matches, up to $n^{\epsilon}$ factors for an arbitrarily small $\epsilon$, the best-known refutation algorithms based on spectral methods \cite{AOW15,BarakM16, dOrsiT23}. %\pravesh{check quantiative bounds here and add refs.}

\paragraph{The case of odd $k$}  For random $k$-XOR refutation when $k$ is odd, O'Donnell and Schramm prove a weaker bound of $m\geq n^{\lceil k/2 \rceil + \epsilon}$ for the same guarantee that is off of the optimal $n^{k/2}$ by a fixed (i.e., that does not improve as $\epsilon \rightarrow 0$) polynomial factor in $n$ . This is because the argument in~\cite{OS19} essentially treats an odd $k$-XOR formula as a $(k+1)$-XOR and then applies their result for even arity XOR.  %\pravesh{Is this correct?}\andrew{yes}%\pravesh{Check if this is accurate} 

As an aside, such a gap between the case of even and odd $k$s has been a fixture in this line of work. The case of odd $k$, for which the above bound appears suboptimal, has indeed proved to be more difficult to tackle in all previous analyses~\cite{Coja-OghlanGL07,AOW15,AbascalGK21,HsiehKM23,GuruswamiKM22} even with more powerful classes of spectral methods.

In the work of O'Donnell and Schramm~\cite{OS19}, the polynomial factor ``discrepancy" can be blamed on the lack of a Cauchy-Schwarz inequality for Sherali-Adams pseudo-distributions. Indeed, the authors themselves remark (see footnote on page 4): "\emph{This is because we do not know how to relate the Sherali–Adams
value of the objective function to its square (local versions of the Cauchy-Schwarz argument result in a loss). Such
a relation would allow us to apply our techniques immediately to prove that Sherali Adams matches the SOS and spectral performance for odd as well as even $k$.}"

\paragraph{Our Results} In this work, we prove an approximate Cauchy-Schwarz inequality that holds for Sherali-Adams pseudo-distributions (so a ``local" version as in the above quote) but with an additive loss that becomes negligible as the locality/degree of the Sherali-Adams LP increases. As an application, we get the right conjectured bound for refuting random $k$-XOR formulas for odd $k$. Along the way, we also observe a natural generalization of the ideas of~\cite{OS19} to the problem of semirandom $k$-XOR refutation where the set of tuples participating in constraints is arbitrary (i.e., worst-case) while the ``right hand sides" are chosen independently at random. 

Let us begin with our approximate Cauchy-Schwarz inequality. For any functions $f,g$ of $x = (x_1, x_2,\ldots, x_n)$ and any distribution $D$ over $x$, $\E_{x \sim D}[ f(x) g(x)] \leq \E[ f(x)^2]^{1/2} \E[ g(x)^2]^{1/2}$. The proof crucially relies on the positive semidefiniteness of the \emph{moment} matrix of the pseudo-distribution. It provably does not hold for Sherali-Adams pseudo-distributions that only satisfy \emph{local} positive semidefiniteness. Indeed, in \cref{sec:sherali-adams-lower-bound}, we observe the following seemingly drastic failure of the Cauchy-Schwarz inequality. 

\begin{theorem}[Sherali-Adams Cannot Satisfy Cauchy-Schwarz Exactly, see \cref{thm:sa-matching-lower-bound}]
For any $R$ and $n \gg R$, there exists an $R$-local Sherali-Adams pseudo-distribution on $\{-1,1\}^n$ such that for the linear polynomial $p(x) = \sum_i x_i$, we have $\pE_D[ p(x)^2] = 0$ and $\pE_{D}[ p] \geq (1 - o_R(1))\frac{n}{\sqrt{R-1}}$.
 \label{thm:sa-lower-bound-main}
\end{theorem}

Despite this failure, we show that there is a natural way to formulate a correct, approximate Cauchy-Schwarz inequality. To the best of our knowledge, such an approximate formulation has not appeared in prior works. Notice that the above counter-example rules out an approximate Cauchy-Schwarz inequality with a \emph{multiplicative} error. Our formulation has a bound on the \emph{additive} error instead.

\begin{theorem}[Approximate Cauchy-Schwarz Inequality]
Let $f,g: \{-1,1\}^n \rightarrow \R$ be polynomials of degree $\leq d$. Let $\|f\|_1, \|g\|_1$ be the L1 norms of the coefficients of $f,g$ respectively. Then, for every $R$-local pseudo-distribution $\cD$ on $\{-1,1\}^n$ and any positive integers $s,t$ such that $d(s+t) \leq R$, we have:
\[
\left|\pE[ fg]\right| \leq \sqrt{ \left(1-\frac{1}{s}\right)\pE[ f^2] + \frac{1}{s} \|f\|_1^2} \sqrt{ \left(1-\frac{1}{t}\right)\pE[ g^2] + \frac{1}{t}\|g\|_1^2}.
\] \label{thm:approxcs-main}
\end{theorem}
\begin{remark}
Note that $\|f\|_1$ and $\|g\|_1$ are basis dependent quantities. This is to be expected as Sherali-Adams (unlike sum-of-squares) is not basis independent. Next, our bounds are written in the $\{\pm 1\}$ basis. They can, of course, be transformed into the $\{0,1\}$ basis, but the quantities appearing will change. Finally, we present the inequality with $s,t$ as independent parameters (as opposed to  setting them both to be $R/2$, say) since we anticipate that in applications, $\|f\|_1, \|g\|_1$ could be different and the freedom to modulate $s,t$ accordingly may lead to better bounds. 
\end{remark}

As we observe in Section \ref{sec:sherali-adams-lower-bound}, \cref{thm:sa-lower-bound-main} shows that the formulation we prove in \cref{thm:approxcs-main} is tight down to the exact constant as $n/R \rightarrow \infty$.

\paragraph{Applications:} We apply Approximate Cauchy-Schwarz to prove tight bounds for CSP refutation.
\begin{theorem}(Sherali-Adams Refutation of Random and Semirandom $k$-XOR, see \Cref{theorem:randomxor} and \Cref{theorem:semirandomxor}) Suppose we have a random or semirandom $k$-XOR instance on $n$ variables with $n^{k/2+\delta}$ expected (in the random case) or total (in the semirandom case) constraints. Then with high probability, there exists some $R=O_{\eps, \delta, k}(1)$ such that for all $R$-local Sherali-Adams pseudodistributions $D$ on $\{-1,1\}^n$, we have $\pE_D \psi(x)\leq 2\eps |\cH|$.
\end{theorem}
In the case of odd $k$, we use our approximate Cauchy-Schwarz lemma to show that the standard spectral certificate used to refute random odd $k$-XOR instances in \cite{AOW15} implies a Sherali-Adams refutation, which allows us to refute with just $n^{k/2+\delta}$ constraints. In addition, we show that we can refute semirandom $k$-XOR in Sherali-Adams, matching the SOS algorithms from \cite{AbascalGK21}. This was previously unknown for both even and odd $k$.

\section{Sherali-Adams Proof of an Approximate Cauchy-Schwarz Inequality}

In this section, we show that Sherali-Adams pseudo-distributions satisfy a natural approximate Cauchy-Schwarz inequality. In the next section, we will see how this inequality helps LP-based refutations of random constraint satisfaction problems. We formulate the approximate Cauchy-Schwarz inequality as follows:

\begin{theorem}[Approximate Cauchy-Schwarz Inequality]
Let $f,g: \{-1,1\}^n \rightarrow \R$ be polynomials of degree $\leq d$. Let $\|f\|_1, \|g\|_1$ be the L1 norms of the coefficients of $f,g$, respectively, when written in the standard monomial basis over $\{-1,1\}^n$. Then, for every $R$-local pseudo-distribution $\cD$ on $\{-1,1\}^n$, and any positive integers $s,t$ such that $d(s+t) \leq R$, we have:
\[
\left|\pE[ fg]\right| \leq \sqrt{ \left(1-\frac{1}{s}\right)\pE[ f^2] + \frac{1}{s} \|f\|_1^2} \sqrt{ \left(1-\frac{1}{t}\right)\pE[ g^2] + \frac{1}{t}\|g\|_1^2}.
\] \label{thm:approxcs}
\end{theorem}

Our proof is based on a simple sampling scheme that allows us to write $f$ and $g$ as averages over \emph{local} functions with some additive ``sampling" error. 

\begin{proof}

We present the proof for $d =1$. The case of larger $d$ is the same but, unlike below, where we ``sample" monomials of degree $1$, we sample monomials of degree $\leq d$ instead.

Let $f(x) = a_0 + \sum_{i \leq n} a_i x_i$ and $g(x) = b_0 + \sum_i b_i x_i$. We will assume that $0< \|f\|_1 = \sum_i |a_i|$ and $0< \|g\|_1 = \sum_i |b_i|$ (if either of these quantities is $0$ then there is nothing to prove). 

We now sample an unbiased estimator for $f$. Define a distribution $p$ on $\{0,1,2,\cdots, n\}$ by setting $p_i = \frac{|a_i|}{\|f\|_1}$. For each $i \in \{0,1,2,\cdots,n\}$, let $\chi_i(x) = x_i$ if $i>0$ and $\chi_0(x) = 1$. Draw $i_1, i_2, \ldots, i_s$ independently from $p$ and define:
\[
F = \frac{1}{s} \sum_{ \ell = 1}^s \frac{a_{i_{\ell}}}{p_{i_{\ell}}} \chi_{i_\ell}(x).
\]
Then, clearly, $F$ depends on at most $s$ variables. Further, $\E[ F] = f$, where the expectation is with respect to $p$. Similarly, construct $G$ from $g$ using $t$ samples. Now, let $\cD$ be any $R$-local pseudo-distribution on $x_1, x_2, \ldots, x_n \in \{\pm 1\}$. Since $F$ and $G$ are $s,t$-local respectively and $s+t \leq R$, we must have:

\begin{equation} \label{eq:local-cs}
| \pE_D[ FG ]| \leq \pE_{D}[F^2]^{1/2} \pE_{D}[ G^2]^{1/2}.
\end{equation}

Now, since $F$ and $G$ are sampled independently, we also have:
\[
\E_{\mathsf{sampling}} \pE[FG] = \pE[ fg].
\]
Thus, using \eqref{eq:local-cs}, taking expectations with respect to the sampling process and applying Jensen's inequality for the expectation of the square root function gives 
\begin{equation} \label{eq:to-be-computed}
\left|\pE[fg]\right| \leq \E_{\mathsf{sampling}} \pE_{\cD}[F^2]^{1/2} \pE_{\cD}[ G^2]^{1/2} \leq \sqrt{ (\E_{\mathsf{sampling}} \pE_{\cD}[F^2])} \sqrt{ \E_{\mathsf{sampling}} \pE_{\cD}[ G^2]}.
\end{equation}

We finish by computing the RHS. 
Let $X_{\ell} = \frac{a_{i_{\ell}}}{p_{i_{\ell}}} \chi_{i_\ell}(x)$, so $F = \frac{1}{s} \sum_{\ell = 1}^s X_{\ell}$. For two distinct samples $\ell \neq m$, we have:
\[
\E_{i_{\ell}, i_m}  [ \pE\left[ X_{\ell} X_{m}\right]] = \pE[f^2].
\]
On the other hand, when $\ell = m$, we have:
\[
\E_{i_{\ell}}  [ \pE[ X_{\ell}^2]] = \sum_i \frac{a_i^2}{p_i} = \left(\sum_i |a_i|\right)^2 = \|f\|_1^2.
\]
Thus, 
\[
\E_{\mathsf{sampling}}[ \pE[ F^2] ] = \frac{s(s-1)}{s^2} \pE[ f^2] + \frac{s}{s^2} \|f\|_1^2 = \left(1-\frac{1}{s}\right) \pE[f^2] + \frac{\|f\|_1^2}{s}.
\]
A similar argument for $G$ gives:
\[
\E_{\mathsf{sampling}}[ \pE[ G^2] ] = \left(1-\frac{1}{t}\right) \pE[g^2] + \frac{\|g\|_1^2}{t}.
\]
Plugging back in \eqref{eq:to-be-computed} finishes the proof.\end{proof}

\subsection{A Matching Lower Bound} \label{sec:sherali-adams-lower-bound}
In this section, we prove that the additive error appearing in our approximate Cauchy-Schwarz inequality is necessary and, in fact, that our bound is quantitatively optimal. 

\begin{theorem} \label{thm:sa-matching-lower-bound}
Let $n,R\geq 2$ be integers such that $n/R \rightarrow \infty$. Then, there exists an $R$-local Sherali-Adams pseudo-distribution on $\{-1,1\}^n$ such that for the linear polynomial $p(x) = \sum_i x_i$, we have
$\pE_D[ p(x)^2] = 0$ and
\[
\pE_{D}[ p] \geq \left(1 - O\left(1/R + R/n\right)\right)\frac{n}{\sqrt{R-1}}.
\]
\end{theorem}
\begin{remark}
Applying \Cref{thm:approxcs} with $f = p$, $g= 1$, $s=R-1$ and $t=1$ gives $\pE_D[p] \leq \frac{n}{\sqrt{R-1}}$. Thus, the additive error in \Cref{thm:approxcs} is asymptotically tight with the correct leading constant. Notice also that this example rules out an approximate Cauchy-Schwarz inequality with any finite \emph{multiplicative} error. 
\end{remark}

\begin{proof}

We will define the pseudo-distribution  $D$ by describing a local distribution on every collection $T$ of $R$ variables. The local distributions will be permutation-invariant on $T$ and thus, the marginal distribution on any $T' \subseteq T$ will depend only on $|T'|$. In particular, local consistency holds: for any $T' \subseteq T_1\cap T_2$ for two different sets of size $R$ each, the marginal of the local distributions on $T_1$ and $T_2$ on $T'$ agree. This immediately implies the existence of an $R$-local Sherali-Adams pseudo-distribution with pseudo-moments agreeing with the local distributions we will construct. 

For every $T \subseteq [n]$ of size exactly $R$, let $Y_T = \sum_{i \in T} x_i$. We will sample $Y_T$ from a distribution on $[-R,R]$ and then set the constituent $x_i$s to be uniform conditioned on satisfying $\sum_i x_i$ being equal to the sampled value of $Y_T$, ensuring the desired permutation invariance. We choose the distribution carefully to have prescribed first and second moments (which will govern $\pE_D [p]$ and $\pE_D[p^2]$, respectively).

Let $A = \frac{R(n-R)}{n-1}$. This will be our prescribed 2nd moment for $Y_T$. Notice that $A = R(1-O(R/n))$ and $\sqrt{A} = \sqrt{R}(1-O(R/n))$. Choose $a$ to be the largest positive integer in $\{-R, -R+2, \cdots, R-2,R\}$ such that $a \leq \sqrt{A}$. Set $b = a+2$. Choose $\theta$ so that $(1-\theta) a^2 + \theta b^2 = A$, that is, set $\theta = \frac{A-a^2}{b^2-a^2}$. Observe that since $a^2 \leq A \leq b^2$, $\theta \in [0,1]$. We construct a distribution for $Y_T$  using the following process: with probability $1-\theta$, output $a$, with probability $\theta$, output $b$. Choose $x_T \in \{-1,1\}^T$ to be uniform over all assignments satisfying $\sum_i x_i = Y_T$. 

We now compute $\pE_D[ x_i x_j]$ for $i \neq j$. Observe that $Y_T^2 = R + 2 \sum_{i<j, i,j\in T} x_i x_j$. Since $\E[ Y_T^2] = A$, we must have $\pE_D[x_i x_j] = \frac{(A-R)}{R(R-1)} = -\frac{1}{n-1}$. Therefore $\pE_D[p^2] = 0$, so a simple computation gives $\pE_D[x_i] = \pE_D[Y_T]/R = ((1-\theta)a + \theta b)/R = \frac{1}{\sqrt{R}} (1- O(1/R + R/n))$. 

Thus, $\pE_D[p] = (1- O(1/R + R/n)) \frac{n}{\sqrt{R}} = (1- O(1/R + R/n)) \frac{n}{\sqrt{R-1}}$, as desired. 
\end{proof}
\newcommand{\psil}{\psi^{(slift)}}
\newcommand{\psic}{\psi^{(cross)}}
\newcommand{\psilc}{\psi^{(L,cross)}}
\newcommand{\psirc}{\psi^{(R,cross)}}
\newcommand{\psid}{\psi^{(diag)}}
\newcommand{\psis}{\psi^{(sq)}}
\newcommand{\phil}{\phi^{(slift)}}
\newcommand{\dps}{d^\phi}
\newcommand{\dph}{d^\phi}
\newcommand{\dih}{d_i^\cH}
\newcommand{\dhh}{d^\cH}
\newcommand{\liftt}[1]{\text{LIFT}\left(#1\right)}
\newcommand{\olift}[1]{\text{oLIFT}\left(#1\right)}
\newcommand{\slift}[1]{\text{sqLIFT}\left(#1\right)}
\newcommand{\xlift}[1]{\text{xLIFT}\left(#1\right)}
\section{Refuting Random $k$-XOR}
\newcommand{\chij}{\cH_i^{(j)}}
In this section, we will show that Sherali-Adams can refute random $k$-XOR instances, which we define below, as an application of approximate Cauchy-Schwarz.
\begin{definition}
    Let $k\geq 2$ be an integer. We say that $\psi$ is a random $k$-XOR instance with $pn^k$ expected clauses if for all $U\in [n]^k$, we add $U$ to $\cH$ with probability $p$ and sample an independent Rademacher ($\pm 1$) sign $b_U$. Then we write $\psi(x)=\sum_{U\in \cH}b_U\cdot x^U$, where $x^U\coloneqq \prod_{i\in U}x_i$.
\end{definition}
\begin{remark}
    Notice that if a $\frac{1}{2}+\gamma$ fraction of constraints are satisfied, i.e. we have $b_U\cdot x^U=1$ for $\frac{1}{2}+\gamma$ fraction of $U\in \cH$, then $\psi(x)=2\gamma|\cH|$. Therefore, certifying that we cannot satisfy a large fraction of constraints is equivalent to certifying an upper bound on $\psi(x)$.
\end{remark}
We also define the following notation that we will use:
\begin{definition}
    Let $\psi=\sum_{U\in \cH}b_U\cdot x^U$ be a $k$-XOR instance with clause set $\cH$. For all $U\not\in\cH$, we let $b_U=0$. Define $\cH_i\coloneqq\{C\in \cH:C_k=i\}$ be the subset of clauses in $\cH$ whose last coordinate is $i$.% and write $\psi_i(x)\coloneqq \sum_{U\in \cH_i}b_U\cdot x^U$.
\end{definition}
We will prove the following theorem:
\begin{theorem}\label{theorem:randomxor}
    Let $k$ be an odd integer and $\psi(x)=\sum_{U\in \cH} b_U\cdot x^U$ be a random $k$-XOR instance with $n^{k/2+\delta}$ expected constraints. Suppose we have $\delta<\frac12$ and $2^{O(k)}\log^3n\leq n^{\delta/2}$. Then for all $\Omega(1)=\gamma\leq 1/2$, there exists some $R=O_{\gamma, \delta, k}(1)$ such that every $R$-local pseudodistribution $\cD$ satisfies $\pE_\cD\psi(x)\leq 2\gamma |\cH|$.
\end{theorem}
\begin{remark}
    Using a standard Fourier analytic reduction, we can generalize this result to general $k$-CSPs.
\end{remark}
The state-of-the-art refutation algorithms for random (and semirandom) $k$-XOR instances are based on spectral certificates. Informally, the idea is to show that the ``value" polynomial of any $k$-XOR instance (the polynomial that computes the advantage over $1/2$ of an assignment) can be upper bounded at every $x$ by a quadratic form of a certain matrix built from the instance. One then establishes that when the input instance is random (or semirandom) with appropriately many constraints, then, the spectral norm of this matrix is small. We will call a matrix satisfying sufficient conditions for obtaining a Sherali-Adams refutation a spectral certificate.
\begin{definition}[Spectral Certificates]\label{def:randcertificate}
Let $k$ be an odd integer and $\psi(x)=\sum_{U\in \cH} b_U\cdot x^U$ be a $k$-XOR instance. Define $A(\psi)$ to be a $n^{k-1}\times n^{k-1}$ matrix defined as follows: For all $(S_1,T_1)\neq (S_2,T_2)\in [n]^{(k-1)/2}\times [n]^{(k-1)/2}$, we have
    \begin{align*}
        A(\psi)_{(S_1,S_2),(T_1,T_2)}\coloneqq \sum_{i\in [n]}b_{S_1T_1i}b_{S_2T_2i}+b_{T_1S_1i}b_{T_2S_2i}\mcom
    \end{align*}
    and $A(\psi)_{(S_1,S_1),(S_2,S_2)}=0$. We call $A(\psi)$ a $\delta$-\emph{spectral certificate for $\cH$} if
    \begin{align*}
        \Norm{A(\psi)}\leq n^{-\delta}  \left(\min_{(S_1,S_2)}\Norm{A(\psi)_{(S_1,S_2)}}_1\right)\mper
    \end{align*}
    We let $D(\psi)_{S_1,S_2}\coloneqq\Norm{A_{S_1,S_2}}_1$ and write $D(\psi)\coloneqq \diag\left(\left\{D(\psi)_{S_1,S_2}\right\}\right)$. Observe that
    \begin{align*}
        x^{\odot (k-1)}A(\psi)x^{\odot (k-1)}=2\sum_i\sum_{U\neq U'\in \cH_i}b_U b_{U'}x^Ux^{U'}\mper
    \end{align*} 
    \end{definition}
With high probability over the choice of $\psi$, the matrix $A(\psi)$ is a spectral certificate. The proof of this follows from the work of \cite{AOW15} along with a short technical argument, and is deferred to the appendix.
\begin{restatable}{lemma}{randomcertificateexists}\label{lem:randomcertificateexists}
    Let $k$ be an odd integer and suppose $\psi$ is a random $k$-XOR instance with $n^{k/2+\delta}$ constraints such that $\delta<\frac12$ and $2^{O(k)}\log^3n\leq n^{\delta/2}$. Then $A(\psi)$ is a $\frac{\delta}{2}$-spectral certificate for $\psi$ with high probability.
\end{restatable}
\begin{remark}
    This result holds for all $\delta<\frac{k}{2}$ but requires significantly more technical details.
\end{remark}
One can view the main idea of O'Donnell and Schramm's work as showing (under appropriate conditions) that spectral bounds on certifying matrices such as above can be proven using the Sherali-Adams linear program. We will abstract their analysis into the following lemma:  
\begin{restatable}{lemma}{saspectralbd}[Abstracting out ~\cite{OS19}]\label{lemma:os19corollary}  
    Let $A$ be a symmetric matrix with integer entries with rows and columns indexed by $t$-tuples of $[n]$, and let $D$ be its degree matrix. Let $\pi_*=\frac{1}{\sum_{S,T}|A_{S,T}|}\min_{S}\Norm{A_S}_1$, $\rho(\cdot)$ denote the spectral radius of a matrix, and $\eta > \min(\pi_*^{-1/2} ,\rho(D^{-1}A))$. Take $\ell \coloneqq \Big\lceil \frac{1}{4} \frac{\log (4\eta^2 \pi_*/25)}{\log (5\rho(D^{-1}A)/2\eta)}\Big\rceil$ and $p \coloneqq \lceil (\frac{5}{2\eta})^{2\ell}\rceil$.
Then for $R(A) \coloneqq tp\ell + t$, any $R$-local pseudoexpectation $\cD$ on $x_1, \cdots, x_n$ satisfies $\left|\pE_\cD (x^{\odot t})^\top Ax^{\odot t}\right|\leq \eta \sum_{S,T}|A_{S,T}|$, where $x^{\odot t}$ is the $n^t$-dimensional vector whose entries are indexed by $t$-tuples of $[n]$ and are equal to the product of the corresponding entries.
\end{restatable}\begin{lemma}\label{theorem:randomrefutation}
    Let $k$ be an odd integer and $\psi(x)=\sum_{U\in \cH}b_U\cdot x^U$ be a $k$-XOR instance such that $|\cH|\geq \frac{36n}{\gamma^2}$ and $|\cH_i|\leq \frac{|\cH|}{n}(1+\gamma)$ for all $i$. Suppose that $A(\psi)$ is a $\delta$-spectral certificate for $\psi$. Then for all $\Omega(1)=\gamma\leq \frac12$, $R=O_{\gamma, \delta, k}(1)$-rounds of Sherali-Adams can certify that $\psi(x)\leq 2\gamma|\cH|$; in other words, for all $R$-local pseudodistributions $\cD$, $|\pE_\cD\psi(x)|\leq 2\gamma |\cH|$.  
\end{lemma}
\begin{proof}
    Let $\psi_i=\sum_{U\in \cH_i}b_U x^U$, $R\coloneqq \max\{R(A(\psi)), \frac{10k(1+\gamma)^2}{\gamma^2}\}=O_{\gamma, \delta, k}(1)$, where $R(A(\psi))$ is the quantity from \Cref{lemma:os19corollary} when $\eta=\gamma^2/2$, and $\pE$ be any $R$-local pseudodistribution. A $\delta$-spectral certificate $A(\psi)$ satisfies $\rho(D^{-1}(\psi)A(\psi))\leq n^{-\delta}$ (recall that $\rho(\cdot)$ is the spectral radius of a matrix) and therefore by \Cref{lemma:os19corollary},
    \begin{align*}
        \pE\sum_i \psi_i^2(x) 
        &= \pE\sum_{U\in \cH}b_U^2 + \frac12\pE x^{\odot {(k-1)}} A(\psi)x^{\odot (k-1)}
        \\&\leq |\cH| + \frac{\gamma^2}{2}\sum_i|\cH_i|(|\cH_i|-1) \leq |\cH|+\frac{\gamma^2}{2}\frac{|\cH|^2}{n}(1+\gamma)^2\mcom\end{align*}
    so by \Cref{thm:approxcs}, we have for $s\geq \frac{9(1+\gamma)^2}{\gamma^2}$,
    \begin{align*}
        \left|\pE \psi(x)\right|=\left|\pE \sum_i \psi_i(x)\right|
        &\leq \sum_i \sqrt{\left(1-\frac1s\right)\pE\psi_i^2+\frac1s\Norm{\psi_i}_1^2}
        \\&\leq \sqrt{n\left(\left(1-\frac1s\right)\sum_i \pE \psi_i^2+\frac1s\sum_i|\cH_i|^2\right)}
        \\&\leq \sqrt{n|\cH|}+\gamma(1+\gamma)|\cH| + \frac{|\cH|}{\sqrt{s}}(1+\gamma)\leq 2\gamma|\cH|\mper
    \end{align*}
\end{proof}
\begin{proof}[Proof of \Cref{theorem:randomxor}]
    For random $k$-XOR, Chernoff bounds imply that $|\cH_i|\leq \frac{|\cH|}{n}(1+\gamma)$ for all $i$ with high probability. Then the result follows directly from \Cref{lem:randomcertificateexists} and \Cref{theorem:randomrefutation}.
\end{proof}

\section{Refuting Semirandom $k$-XOR}
In this section, we show that Sherali-Adams can refute semirandom instances of $k$-XOR.
\begin{definition}
    Let $k$ be an integer and $\cH\subseteq [n]^k$. We say that $\psi$ is a semirandom $k$-XOR instance with clause set $\cH$ if for $U\in \cH$, we sample an independent Rademacher ($\pm 1$) sign $b_U$ and write $\psi(x)=\sum_{U\in \cH}b_U\cdot x^U$.
\end{definition}
\begin{theorem}\label{theorem:semirandomxor}
    Let $k$ be an integer. $\psi(x)=\sum_{U\in \cH} b_U\cdot x^U$ be a semirandom $k$-XOR instance with clause set $\cH$ such that $|\cH|=n^{k/2+\delta}$ total  constraints. Suppose we have $5n^{-\delta/36}\leq\gamma\leq \frac12$ and $16384k^4\log^3 n\leq \gamma^{1/18}n^{\delta/18}$. Then with high probability, there exists some $R=O_{\gamma, \delta, k}(1)$ such that every $R$-local pseudodistribution $\cD$ satisfies $\pE_\cD\psi(x)\leq 2\gamma |\cH|$.
\end{theorem}
For semirandom odd $k$-XOR, the $\cH_i$ can be much larger than $\frac{|\cH|}{n}$, in which case we need to partition large $\cH_i$ into smaller sets to ensure the matrix $A(\psi)$ does not have too many nonzero entries. In addition, we can no longer guarantee that $A(\psi)$ is a spectral certificate with high probability. We thus define a notion of a regular partition and a refined notion of a spectral certificate that allows us to delete a negligible fraction of rows and columns of the matrix.
\begin{definition}
    For all $\cH=\cH_1\sqcup\cdots\sqcup\cH_n\subseteq [n]^k$,
    we call a partition $P\coloneqq\bigsqcup_i P_i\coloneqq\bigsqcup_i\bigsqcup_{j\in [p_i]}\chij$ of the $\cH_i$s \emph{regular} if it has a total of at most $2n$ parts and $|\chij|\leq \frac{|\cH|}{n}$ for all $i,j$.
\end{definition}
\begin{definition}[$(\gamma,\delta)$-Spectral Certificates]
Let $k$ be an odd integer and $\psi(x)=\sum_{U\in \cH} b_U\cdot x^U$ be a $k$-XOR instance with clause set $\cH$. Suppose we have a partition $P_i\coloneqq\sqcup_{j\in [p_i]}\chij$ of each $\cH_i$, and let $P=(P_1, \cdots, P_n)$. 
    Define $A'_P(\psi)$ to be a $n^{k-1}\times n^{k-1}$ matrix defined as follows: For all $(S_1,T_1)\neq (S_2,T_2)\in [n]^{(k-1)/2}\times [n]^{(k-1)/2}$, we have
    \begin{align*}
        A'_P(\psi)_{(S_1,S_2),(T_1,T_2)}\coloneqq \sum_{i\in [n]}\sum_{j\in [p_i]} b_{S_1T_1i}b_{S_2T_2i}\cdot \1\{S_1T_1i,S_2T_2i\in \chij\}+b_{T_1S_1i}b_{T_2S_2i}\cdot \1\{T_1S_1i,T_2S_2i\in \chij\}\mcom
    \end{align*}
    and $A'_P(\psi)_{(S_1,S_1),(S_2,S_2)}=0$. Let $A(\psi)$ be a principal submatrix of $A'_P(\psi)$, and define $D(\psi)_{S_1,S_2}\coloneqq\Norm{A_{S_1,S_2}}_1$ and write $D(\psi)\coloneqq \diag\left(\left\{D(\psi)_{S_1,S_2}\right\}\right)$. We call $A(\psi)$ a \emph{$(\gamma,\delta)$-spectral certificate of $\psi$ with respect to $P$} if the following hold: 
    \begin{itemize}
        \item $A(\psi)$ can be formed by removing a collection of rows of $A'_P(\psi)$ and their corresponding columns whose $\ell_1$ norms sum to at most $\frac{\gamma ^2}{2}\sum_{i\in [n]}\sum_{j\in [p_i]}|\chij|(|\chij|-1)$,
        %\item $\Norm{A(\psi)}\leq n^{-\delta}\left(\min_{S_1,S_2}\Norm{A_{S_1,S_2}}_1\right)$.
        \item $\rho(D^{-1}(\psi)A(\psi))\leq n^{-\delta/2}$, where $\rho(\cdot)$ denotes the spectral radius of a matrix.
    \end{itemize}
Observe that (abusing notation and treating $A(\psi)$ as a $n^{k-1}\times n^{k-1}$ matrix where the removed rows are replaced by zeros), we have
    \begin{align*}
        x^{\odot (k-1)}A'_P(\psi)x^{\odot (k-1)}=2
        \sum_i\sum_{j\in [p_i]}\sum_{U\neq U'\in \chij}b_U b_{U'}x^Ux^{U'}\mcom
    \end{align*} 
    and 
        \begin{align*}
        x^{\odot (k-1)}\left(A'_P(\psi)-A(\psi)\right)x^{\odot (k-1)}\leq\frac{\gamma ^2}{2}\sum_{S_1,S_2}|A'_P(\psi)_{S_1,S_2}|\mper
    \end{align*}
    \end{definition}
    \newcommand{\psio}{\psi^\1}
\begin{definition}\label{def:psio}
    Let $k$ be an integer and $\psi(x)=\sum_{U\in \cH}b_U\cdot x^U$ be a $k$-XOR instance with clause set $\cH$. Then we define $\psio(x)\coloneqq \sum_{U\in \cH}x^U$ to be the $k$-XOR instance with the same constraints but all ones on the RHS.
\end{definition}
\begin{remark}\label{remark:allonescomparison}
    Note that any entry $A'_P(\psi)_{(S_1,S_2),(T_1,T_2)}$ is the sum of $A'_P(\psio)_{(S_1,S_2),(T_1,T_2)}$ independent Rademacher random variables. In particular, if $A'_P(\psio)_{(S_1,S_2),(T_1,T_2)}\neq 0$, then $A'_P(\psi)_{(S_1,S_2),(T_1,T_2)}=0$ with probability at most $\frac12$.
\end{remark}
We will show that the existence of a $(\gamma,\delta)$-spectral certificate of $\psi$ implies that Sherali-Adams admits a refutation of $\psi$. 
\begin{lemma}\label{theorem:semirandomrefutation}
    Let $k$ be an odd integer and $\psi(x)=\sum_{U\in \cH}b_U\cdot x^U$ be a semirandom $k$-XOR instance. Suppose that there exists a  $(\gamma,\delta)$-spectral certificate for $\psi$ with respect to a regular partition $P$ of all $\cH_i$. Then if $\gamma\leq \frac12$ and $|\cH|\geq\frac{4n}{\gamma^2}$, then $R=O_{\gamma, \delta, k}(1)$-rounds of Sherali-Adams can certify that $\psi(x)\leq 2\gamma|\cH|$; in other words, for all $R$-local pseudodistributions $\cD$, $|\pE_\cD\psi(x)|\leq 2\gamma |\cH|$. 
\end{lemma}
\begin{proof}[Proof of \Cref{theorem:semirandomrefutation}]
    Let $\psi_{ij}=\sum_{U\in \chij}b_U x^U$ and $A(\psi)$ be a $(\gamma, \delta)$-spectral certificate of $\psi$. Let  $R\coloneqq \max\{R(A(\psi)), \frac{5k}{\gamma^2}\}=O_{\gamma, \delta, k}(1)$, where $R(A(\psi))$ is the quantity from \Cref{lemma:os19corollary} when $\eta=\frac{\gamma^2}{2}$, and $\pE$ be an $R$-local pseudodistribution. Then $\rho(D^{-1}(\psi)A(\psi))\leq n^{-\delta/2}$ and thus by \Cref{lemma:os19corollary},
    \begin{align*}
        \pE\sum_{i}\sum_{j\in [p_i]} \psi_{ij}^2(x) 
        &= \pE\sum_{U\in \cH}b_U^2 + \frac12\pE x^{\odot {(k-1)}} A'(\psi)x^{\odot (k-1)}
        \\&\leq |\cH| + \frac12\pE x^{\odot {(k-1)}} A(\psi)x^{\odot (k-1)} + \frac12\pE x^{\odot {(k-1)}} (A'(\psi)-A(\psi))x^{\odot (k-1)}
        \\&\leq |\cH| + \left(\frac{\gamma^2}{2}+\frac{\gamma^2}{4}\right)\sum_i\sum_{j\in [p_i]}|\chij|(|\chij|-1) 
        \\&\leq |\cH|+\frac{3\gamma^2}{4}\frac{|\cH|^2}{n}\mcom\end{align*}
    where the last line uses the fact that our partition is regular and therefore each $\chij$ has size at most $\frac{|\cH|}{n}$. Thus by \Cref{thm:approxcs}, since $P$ has at most $2n$ parts, we have, for $s\geq \frac{4}{\gamma^2}$,
    \begin{align*}
        \left|\pE\psi\right|=\left|\pE \sum_i\sum_{j\in [p_i]} \psi_{ij}(x)\right|
        &\leq \sum_i \sum_{j\in [p_i]} \sqrt{\left(1-\frac1s\right)\pE\psi_{ij}^2+\frac1s\Norm{\psi_{ij}}_1^2}
        \\&\leq \sqrt{2n\left(\left(1-\frac1s\right)\sum_i\left(\sum_{j\in [p_i]}  \pE \psi_{ij}^2\right)+\sum_i\frac{|\cH_i|^2}{sn}\right)}
        \\&\leq \sqrt{2n|\cH|+3\gamma^2|\cH|^2+\frac{2|\cH|^2}{s}}
        \leq 2\gamma|\cH|\mper
    \end{align*}
\end{proof}
Finally, we note that given sufficient regularity conditions on $\cH$, there exists a spectral certificate of $\psi$ with high probability.
\begin{restatable}{lemma}{crossspecnorm}
\label{lem:crossspecnorm}
    Let $k$ be odd and $\psi(x)=\sum_{U\in \cH}b_U\cdot x^U$ be a semirandom $k$-XOR instance with constraint set $\cH\subseteq [n]^k$ such that $|\cH|=n^{k/2+\delta}$ and let $P$ be a regular partition of the $\cH_i$s. Suppose $\cH$ has the following properties: 
    \begin{itemize}
        \item For all $S\in [n]^{(k-1)/2}$, $i\in [n]$, there are at most $n^{\delta/2}$ total clauses of the form $STi$ or $TSi$ in $\cH$.
        \item For all $S_1T_1\in [n]^{(k-1)/2}\times [n]^{(k-1)/2}$, there are at most $n^{\delta/2}$ total clauses of the form $S_1T_1i\in \cH$.
    \end{itemize}
    Suppose that $\gamma\geq n^{-\delta/4}\sqrt{60}$ and $16384k^4\log^3 n\leq n^{\delta/18}$. Then with high probability, there exists a $(\gamma,\delta/9)$-spectral certificate of $\psi$ with respect to $P$.\end{restatable}
    We defer the proof of this lemma to \Cref{sec:existence proofs}.

\subsection{Refuting induced $2$-XOR instances}
Unlike the random case, a semirandom XOR instance $\psi$ does not necessarily have a spectral certificate: when large groups of variables occur together in a large number of clauses, the spectral norm of $A'(\psi)$ can be large and making it sufficiently small requires pruning a large (possible almost all) of the rows. Instead, we can partition our set of clauses $\cH$ into two subparts $\cH'$ and $\cH''$, where $\cH'$ is the "approximately regular" part where 
we can construct the aforementioned spectral certificate, and for the second part $\cH''$, we instead construct a $2$-XOR instance that we can refute using Sherali-Adams. At a high level, this resembles the approach of \cite{AbascalGK21}, however, their method for refuting the $2$-XOR instance for $\cH''$ relies on bounding the value of the instance by an infinity-to-1 norm and using the fact that SOS can recognize Grothendieck's inequality. This is not known for Sherali-Adams, so we must handle this case differently. 

We define the notion of an induced $2$-XOR instance below. This is a $2$-XOR instance formed by splitting each clause $U$ of a $k$-XOR instance into two parts $S\sqcup T$, defining a variable $y_S$ and $y_T$ for each of the two parts, and replacing the $k$-XOR constraint $b_U x^U$ with the $2$-XOR constraint $b_U y_S y_T$. Formally, we have a bijective map between our clause $U\in \cH$ and the partition $(S,T)$, and using the transformation $y_S=\prod_{i\in S}x_i$, for all $x$, there exists $y$ such that $\phi(y)=\psi(x)$.
\begin{definition}
    Let $\psi(x)=\sum_{U\in \cH}b_U\cdot x^U$ be a semirandom $k$-XOR instance with constraint set $\cH\subseteq[n]^k$. We call a $2$-XOR instance $\phi$ an \emph{$f$-induced $2$-XOR instance of $\psi$} if each of its variables $y_S$ is indexed by some $S\subseteq [n]^{\leq k}$ and there exists a  set $\cB=\cB(\phi)$ of pairs of variables $(S,T)$ and a bijection $f:\cB\rightarrow \cH$ such that $ST$ is a permutation of $f(S,T)$ and $\phi$ can be written as 
    \begin{align*}
        \phi = \sum_{(S,T)\in \cB}b_{f(S, T)} y_S y_T \mcom
    \end{align*}
     Let $\dph(S)=|\{U\in \cH: f^{-1}(U)\ni S\}|$ denote the number of times the variable $y_S$ appears in the above expression for $\phi$, which we call the \emph{degree} of $S$ in $\phi$. Then we define the weighted adjacency matrix $A'(\phi)$ of $\phi$ with rows and columns indexed by variables of $\phi$: 
  \begin{align*}\left(A'(\phi)\right)_{S,T}=b_{f(S,T)}\1\{(S,T)\in \cB\}+b_{f(T,S)}\1\{(T,S)\in \cB\}\mper
 \end{align*}

Let $A(\phi)$ be a principal submatrix of $A'(\phi)$ and let $D(\phi)$ be the degree matrix of $A(\phi)$. We call $A(\phi)$ an \emph{induced $(\gamma,\delta)$-spectral certificate for $\psi$} if the following hold:
\begin{itemize}
    \item $A(\phi)$ can be formed by deleting rows of $A'(\phi)$ and their corresponding columns whose $\ell_1$ norms sum to at most $\gamma\sum_{S,T}|A'(\phi)_{S,T}|$,
    \item $\rho(D^{-1}(\phi)A(\phi))\leq n^{-\delta}$.
\end{itemize}

\end{definition}
A Sherali-Adams refutation of an induced $2$-XOR instance of $\psi$ implies a refutation of $\psi$ itself, which we prove in the appendix.
\begin{restatable}{lemma}{inducedtoorig}
    \label{lem:inducedtoorig}
    Let $\psi(x)=\sum_{U\in \cH}b_U\cdot x^U$ be a $k$-XOR instance with constraint set $\cH\subseteq [n]^k$ and $\phi$ be an $f$-induced $2$-XOR instance of $\psi$. Suppose that for all $R$-local pseudodistributions $\cD'$, we have $|\pE_{\cD'} \phi(y)|\leq 2\gamma |\cH|$. Then for all $kR$-local pseudodistributions $\cD$, we have $|\pE_\cD \psi(x)|\leq 2\gamma |\cH|$.
\end{restatable}
Using this, we prove that Sherali-Adams can refute XOR instances with an induced spectral certificate.
\begin{lemma}\label{lem:refuteinduced}
    Let $\psi(x)=\sum_{U\in \cH}b_U\cdot x^U$ be a $k$-XOR instance and suppose there exists an induced $(\gamma, \delta)$-spectral certificate for $\psi$. Then $R=O_{\gamma,\delta,k}(1)$-rounds of Sherali-Adams can certify that $\psi(x)\leq 2\gamma |\cH|$; in other words, for all $R$-local pseudodistributions $\cD$, $|\pE_\cD \psi(x)|\leq 2\gamma|\cH|$.
\end{lemma}
\begin{proof}
    Let $\phi$ be an $f$-induced $2$-XOR instance of $\psi$ and $A(\phi)$ be a corresponding induced $(\gamma, \delta)$-spectral certificate. Let $R\coloneqq k\cdot R(A(\phi))$, where $R(A(\phi))$ is the quantity from \Cref{lemma:os19corollary} where $\eta=\gamma$, and let $\cD$ be any $R=O_{\gamma,\delta,k}(1)$-local pseudodistribution. Since $\rho(D^{-1}(\phi)A(\phi))\leq n^{-\delta}$, by \Cref{lemma:os19corollary}, for any $R(A(\phi))$-local pseudodistribution $\cD'$, we have
    \begin{align*}
        |\pE_{\cD'}\phi(y)|
        =\frac{1}{2}\left|\pE y^\top A'(\phi)y\right|
        \leq\frac{1}{2}\left|\pE y^\top A(\phi)y\right| + \frac{1}{2}\left|\pE y^T(A'(\phi)-A(\phi))y\right|
        \leq \gamma|\cH|+\gamma|\cH|=2\gamma|\cH|\mcom
    \end{align*}
    so $\pE_\cD \psi(x)\leq 2\gamma|\cH|$ by \Cref{lem:inducedtoorig}. 
\end{proof}
Furthermore, an induced $2$-XOR instances without too many variables admits an induced spectral certificate for its parent instance, which we show in \Cref{sec:existence proofs}.
\begin{restatable}{lemma}{inducedspecnorm}
\label{lem:inducedspecnorm}
    Let $\psi(x)=\sum_{U\in \cH}b_U\cdot x^U$ be a semirandom $k$-XOR instance with constraint set $\cH\subseteq [n]^k$ and $\phi$ be an $f$-induced $2$-XOR instance of $\psi$ with variable set $\cS$ such that $|\cS|\leq n^{-\delta}|\cH|$. If $\gamma\geq n^{-\delta/2}$ and $4\sqrt{k\log n}\leq n^{\delta/20}$, then with high probability, there exists an induced $(\gamma, \frac{\delta}{5})$-spectral certificate for $\psi$ with respect to $(\cB,f)$.
    \end{restatable}
\subsection{Proof of \Cref{theorem:semirandomxor}}
Here we prove the main result of this section.
\begin{proof}[Proof of \Cref{theorem:semirandomxor}]
First, suppose that $k$ is even. For all $U=ST\in \cH$, where $S,T\in [n]^{k/2}$ define $f(S,T)\coloneqq U$. Construct the $f$-induced $2$-XOR instance $\phi$ of $\psi$. Then we have at most $n^{k/2}=n^{-\delta}|\cH|$ total variables, so by \Cref{lem:inducedspecnorm}, there exists an induced $(\gamma, \frac{\delta}{5})$-spectral certificate for $\psi$ with respect to $f$. Thus by \Cref{lem:refuteinduced}, there exists some $R=O_{\gamma, \delta,k}(1)$ such that for all $R$-local pseudodistributions $\cD$, we have $|\pE_\cD \psi(x)|\leq 2\gamma|\cH|$.

Now suppose that $k$ is odd. We construct a partition $\cH=\cH'
\sqcup \cH''\sqcup \cH'''$ as follows:
\begin{enumerate}
    \item While there exists some $S\in [n]^{(k-1)/2}$ and $i\in [n]$ such that there are more than $n^{\delta/2}$ total clauses of the form $STi$ or $TSi$ in $\cH$, move all $STi$ and $TSi$ to a new hypergraph $\cH'$.
    \item Afterwards, while there exist some $(S,T)\in [n]^{(k-1)/2}\times [n]^{(k-1)/2}$ such that there are more than $n^{\delta/2}$ total clauses of the form $STi\in \cH$, move all such $STi$ from $\cH$ to a new hypergraph $\cH''$. 
    \item Let $\cH'''$ consist of the remaining clauses.
\end{enumerate}
Now, consider $\cH'$. If $|\cH'|\leq \frac{\gamma}{2}|\cH|$, then $\psi(\cH')\leq \frac{\gamma}{2}|\cH|$ trivially. Otherwise, each edge $STi\in\cH'$ consists of one part $(S,i)\in [n]^{(k-1)/2}\times [n]$ and one part $T\in [n]^{(k-1)/2}$. By construction, we have at most $n^{k/2+\delta/2}$ different $(S,i)$ that occur in $\cH'$, since each must occur in at least $n^{\delta/2}$ tuples for it to be added to $\cH'$, and there are at most $n^{(k-1)/2}$ total $T$. Thus, when we use these clauses to construct an induced $2$-XOR instance (letting $f((S,i),T)=STi$ if $(S,i)$ was the heavy part and $f(S,(T,i))=STi$ if $(T,i)$ was the heavy part), we will have $\leq n^{k/2+\delta/2}+n^{(k-1)/2}\leq n^{-\delta/3}|\cH'|$ total variables. Thus by \Cref{lem:inducedspecnorm}, there exists an induced $(\frac{\gamma}{2}, \frac{\delta}{15})$-spectral certificate for $\psi(\cH')$, so by \Cref{lem:refuteinduced} there exists some $R'=O_{\gamma/2, \delta/15, k}(1)=O_{\gamma,\delta,k}(1)$ such that for all $R'$-local pseudodistributions $\cD'$, we have $\pE_{\cD'}\psi(\cH')\leq \gamma|\cH'|$.

Similarly for $\cH''$ either $|\cH''|\leq \frac{\gamma}{2}|\cH|$ and thus $\psi(\cH'')\leq \frac{\gamma}{2}|\cH|$, or we can construct an induced $2$-XOR instance whose variables are $ST\in [n]^{k-1}$ which occur together in at least $n^{\delta/2}$ clauses and $i\in [n]$. Therefore we will have $\leq n^{k/2+\delta/2}+n\leq n^{-\delta/3}|\cH''|$ total variables, so like before there exists an induced $(\frac{\gamma}{2}, \frac{\delta}{15})$-spectral certificate for $\psi(\cH'')$ so for all $R'$-local pseudodistributions $\cD'$, we have $\pE_{\cD'}\psi(\cH'')\leq \gamma|\cH''|$.

Finally, for $\cH'''$, if $|\cH'''|\leq \frac{\gamma}{2}|\cH|$, then $\psi(\cH''')\leq \frac{\gamma}{2}|\cH|$. Otherwise, construct a regular partition $P$ of $\cH'''$ by taking each $\cH'''_i$ of size larger than $\frac{|\cH'''|}{n}$ and repeatedly peeling off $\frac{|\cH'''|}{n}$ at a time. Then $\cH'''$ satisfies the conditions of \Cref{lem:crossspecnorm} so there exists a $(\frac{\gamma}{2}, \frac{\delta}{9})$-spectral certificate for $\psi(\cH''')$. Thus by \Cref{theorem:semirandomrefutation}, there exists some $R'''=O_{\gamma/2, \delta/9, k}(1)=O_{\gamma, \delta, k}(1)$ such that for all $R'''$-local pseudodistributions $\cD'$, we have $\pE_{\cD'}\psi(\cH''')\leq \gamma |\cH'''|$.

Combining the above and letting $R\coloneqq \max(R',R''')$, for all $R$-local pseudodistributions $\cD$, we have
\begin{align*}
    \left|\pE_\cD\psi(x)\right|
    &\leq \left|\pE_\cD \psi(\cH')(x)\right|+\left|\pE_\cD \psi(\cH'')(x)\right|+\left|\pE_\cD \psi(\cH''')(x)\right|
    \\&\leq \frac{\gamma}{2}|\cH|+\frac{\gamma}{2}|\cH|+\gamma(|\cH'|+|\cH''|+|\cH'''|)
    \\&\leq 2\gamma|\cH|\mper
\end{align*}
\end{proof}

\subsection{Existence of spectral certificates}
\label{sec:existence proofs}
Here we prove \Cref{lem:crossspecnorm} and \Cref{lem:inducedspecnorm}. We use the following result, which we defer to the appendix, in the proof of \Cref{lem:crossspecnorm}.
\begin{restatable}{lemma}{canprune}
    \label{lem:canprune}
 Let $\psi(x)=\sum_{U\in \cH}b_U\cdot x^U$ be a semirandom $k$-XOR instance with constraint set $\cH\subseteq [n]^k$ such that $|\cH|=n^{k/2+\delta}$ and let $P$ be a regular partition of the $\cH_i$s. Suppose $\cH$ has the following properties: 
    \begin{itemize}
        \item For all $S\in [n]^{(k-1)/2}$, $i\in [n]$, there are at most $n^{\delta/2}$ total clauses of the form $STi$ or $TSi$ in $\cH$.
        \item For all $S_1T_1\in [n]^{(k-1)/2}\times [n]^{(k-1)/2}$, there are at most $n^{\delta/2}$ total clauses of the form $S_1T_1i\in \cH$.
    \end{itemize}
    Then with high probability, there exists a principal submatrix $A(\psi)$ of $A'(\psi)$ such that all of the following hold:
    \begin{itemize}
        \item The deleted rows and columns of $A'(\psi)$ have $\ell_1$ norms summing to at most $30n^{-\delta/2}\sum_{i\in [n]}\sum_{j\in [p_i]}|\chij|(|\chij|-1)$.
        \item Each row of $A(\psi)$ has at least $n^\delta$ nonzero entries.
        \item Each row of $A(\psi)$ has at least $\frac13$ as many nonzero entries as the corresponding row of $A(\psio)$, where we recall that $A'(\psio)$ is the matrix defined with the same $\cH$ as $\psi$ but with all $b_U$s equal to $1$.
        \item Each entry of $A(\psi)$ has absolute value at most $4kn^{\delta/4}\sqrt{\log n}$.
    \end{itemize}
\end{restatable}
\crossspecnorm*
\begin{proof}
    Let $A(\psi)$ be the principal submatrix of $A'(\psi)$ given by \Cref{lem:canprune}. Since $\gamma\geq \sqrt{60}n^{-\delta/4}$, it satisfies the first condition of a $(\gamma, \delta)$-spectral certificate with high probability. Now we prove the second condition. Let $D(\psi)$ be the degree matrix of $A(\psi)$, $D(\psio)$ be the degree matrix of $A(\psio)$, and $M(\psi)$ be the degree matrix of the support matrix of $A(\psio)$, i.e. a diagonal entry of $M(\psi)$ is the number of nonzero entries in the corresponding row of $A(\psio)$. 
    
    Note that we have $M(\psi)\leq 3D(\psi)$ entrywise. Define $\Gamma\coloneqq M(\psi)^{-1/2}A(\psi)M(\psi)^{-1/2}$, so we have 
    \begin{align*}
        \specrad(D(\psi)^{-1}A(\psi))\leq 3  \rho(M(\psi)^{-1}A(\psi))=3\specrad(\Gamma)\leq 3\Norm{\Gamma}\mper
    \end{align*}
    We start by upper bounding $\E\tr(\Gamma^\ell)=\E \tr((M^{-1}(\psi)A(\psi))^\ell)$ for even $\ell$ by analyzing closed length-$\ell$ walks $V_1\mapsto V_2\mapsto\cdots\mapsto V_\ell\mapsto V_1$ in the weighted multigraph whose vertices are variables of $\psi$ corresponding to the weighted signed adjacency matrix $M^{-1}(\psi)A(\psi)$. Note that for any edge between two vertices $V_j=(S_1,S_2)$ and $V_{j+1}=(T_1,T_2)$, there exists some $i$ such that either both $S_1T_1i,S_2T_2i\in\cH$ or both $T_1S_1i,T_2S_2i\in \cH$. Since each $U\in\cH$ comes with an independent Rademacher sign $\xi_{U}$, if some clause $U\in \cH$ appears only once as one of the two clauses inducing an edge in our walk, then the expected value of the walk is $0$. Therefore, in a walk with a nonzero contribution to the expected trace, each clause which induces an edge must occur at least twice. 
    
    Then we can define an encoding of any closed walk $V_1\mapsto V_2\mapsto\cdots\mapsto V_\ell\mapsto V_1$ by first choosing a variable $V_j=(S_1,S_2)$, of which there are at most $n^{k-1}$ choices, then choosing $b\in \{0,1\}^\ell$ such that $b_j$ is the indicator that both of the two tuples which induce the edge from $V_j$ to $V_{j+1}$ have not induced a previous step in the walk. Since each edge must occur at least twice, we note that $b_j=0$ for at least $\frac{\ell}{2}$ values of $j$.

    At each step $j$, we choose the next step of the walk based on the value of $b_j$: At $V_{j}=(S_1,S_2)$, if $b_j=1$, we simply pick an arbitrary neighbor of $V_{j}$, while if $b_j=0$, there are at most $2j\leq 2\ell$ possible $U\in \cH$ to choose from which have induced a previous step in the walk. Therefore, when $b_j=1$, there are at most $D(\psi)(V_j)$ possible choices, so $\sum_{V_{j+1}\in N_1(V_{j})}\frac{1}{M(\psi)(V_{j})}\leq \frac{D(\psi)(V_j)}{M(\psi)(V_j)}\leq 4kn^{\delta/4}\sqrt{\log n}$. When $b_j=0$, without loss of generality, suppose the tuple we pick is of the form $S_1T_1i\in \cH$ for some $i\in [n]$. Then the second tuple inducing this edge must be of the form $(S_2,T_2,i)\in \chij$ for some $j\in [p_i]$, and there are at most $n^{\delta/2}$ total choices for that, so $\sum_{V_{j+1}\in N_0(V_{j})}\frac{1}{M(\psi)(V_{j})}\leq \frac{2\ell n^{\delta/2}}{n^\delta}\leq 2\ell n^{-\delta/2}$, since $M(\psi)(V_j)\geq n^{\delta}$. Thus we get
    \begin{align*}
        \E\tr(\Gamma^\ell)&\leq\sum_{b\in \{0,1\}^\ell}\sum_{V_1\in [n]^{k-1}}\frac{1}{M(\psi)(V_1)}\sum_{S_2\in N_{b_1}(S_1)}\frac{1}{M(\psi)(S_2)}\cdots \sum_{S_\ell\in N_{b_{\ell - 1}}(S_{\ell - 1})}\frac{1\cdot \mathbb{1}\{S_1\in N_{b_\ell}(S_\ell)\}}{M(\psi)  (S_\ell)}
        \\&\leq 2^\ell n^{k-1}(4kn^{\delta/4}\sqrt{\log n}\cdot 2\ell n^{-\delta/2})^{\ell / 2}
        \\&=2^{\ell}n^{k-1} (2\sqrt{2k\ell}(\log n)^{1/4} n^{-\delta/8})^{\ell}\mper
    \end{align*}
    Therefore, we set $\ell=2\left\lceil \frac{k\log n}{2}\right\rceil$ so by Markov's inequality, we have
    \begin{align*}
        \Pr\left[\Norm{\Gamma}\geq n^{-\delta/9}\geq  8\sqrt2kn^{-\delta/8}\log^{3/4}n \right]\leq \frac{\E\tr(\Gamma^\ell)}{(8\sqrt2kn^{-\delta/8}\log^{3/4}n)^\ell}=\frac{1}{\poly(n^k)}\mper
    \end{align*}
\end{proof}

\inducedspecnorm*
\begin{proof}
While $A'(\phi)$ contains a row with at most $n^{\delta/2}$ entries, delete that row and its corresponding column, and let $A(\phi)$ be the result. Then note that the deleted rows contain at most $|\cS|n^{\delta/2}\leq n^{-\delta/2}|\cH|$ entries. Now let $D(\phi)$ be the degree matrix of $A(\phi)$, and note that all diagonal entries of $D(\phi)$ are at least $n^{\delta/2}$. Define $\Gamma=D(\phi)^{-1/2}A(\phi)D(\phi)^{-1/2}$, and note that $\specrad(D(\phi)^{-1}A(\phi))=\specrad(\Gamma)\leq\Norm{\Gamma}$. We start by upper bounding $\E\tr(\Gamma^\ell)=\E \tr((D^{-1}(\phi)A(\phi))^\ell)$ by analyzing closed even length-$\ell$ walks $S_1\mapsto S_2\mapsto\cdots\mapsto S_\ell\mapsto S_1$ in the weighted graph whose vertices are tuples $(S,S')$ which corresponds to the weighted signed adjacency matrix $D^{-1}(\phi)A(\phi)$. Note that there exists an edge between two vertices $S_i,S_{i+1}$ if either $f(S_i,S_{i+1})\in\cH$ or $f(S_{i+1},S_i)\in \cH$. Since each $U\in\cH$ comes with an independent Rademacher sign $\xi_{U}$, if some unordered tuple $(S,T)$ appears only once as an edge of the walk, then the expected value of the walk is $0$. Therefore, each edge in a walk with a nonzero contribution to the expected trace must occur at least twice. Then we can define an encoding of any closed walk $S_1\mapsto S_2\mapsto\cdots\mapsto S_\ell\mapsto S_1$ by first choosing a variable $S_1$, of which there are $|[n]^{k-1}|<n^k$ choices, then choosing $b\in \{0,1\}^\ell$ such that $b_i$ is the indicator that the $i$th edge of the walk has not occurred in a previous step. Since each edge must occur at least twice, we note that $b_i=0$ for at least $\frac{\ell}{2}$ values of $i$.

    At each step $i$, we choose the next step of the walk based on the value of $b_i$: At $S_{i-1}$, if $b_i=1$, we simply pick an arbitrary neighbor of $S_{i-1}$, while if $b_i=0$, we pick some $S_i$ such that the edge $(S_{i-1},S_i)$ has occurred previously. Then, when $b_i=1$, there are at most $D(\phi)(S_{i-1})$ possible choices, so $\sum_{S_i\in N_1(S_{i-1})}1/D(\phi)(S_{i-1})\leq 1$. When $b_i=0$, there are at most $i-1\leq\ell$ previous edges to choose from, so $\sum_{S_i\in N_0(S_{i-1})}1/D(\phi)(S_{i-1})\leq \ell n^{-\delta/2}$. Combining the above, we get
    \begin{align*}
        \E\tr(\Gamma^\ell)&\leq\sum_{b\in \{0,1\}^\ell}\sum_{S_1\in \cS}\frac{1}{D(\phi)(S_1)}\sum_{S_2\in N_{b_1}(S_1)}\frac{1}{D(\phi)(S_2)}\cdots \sum_{S_\ell\in N_{b_{\ell - 1}}(S_{\ell - 1})}\frac{1\cdot \mathbb{1}\{S_1\in N_{b_\ell}(S_\ell)\}}{D(\phi)  (S_\ell)}
        \\&\leq 2^\ell n^{k}(\ell n^{-\delta/2})^{\ell / 2}\mper
    \end{align*}
    Therefore, we set $\ell=2\left\lceil \frac{3k\log_2 n}{4}\right\rceil$ so by Markov's inequality, we have
    \begin{align*}
        \Pr\left[\Norm{\Gamma}\geq n^{-\delta/5}\geq 4n^{-\delta/4}\sqrt{k\log_2 n} \right]\leq \frac{\E\tr(\Gamma^\ell)}{(4n^{-\delta/4}\sqrt{k\log_2 n})^\ell}=\frac{1}{\poly(n^k)}\mper
    \end{align*}
\end{proof}

\bibliographystyle{alpha}
\bibliography{references}

@inproceedings{OS19,
author = {O'Donnell, Ryan and Schramm, Tselil},
title = {Sherali-adams strikes back},
year = {2019},
isbn = {9783959771160},
publisher = {Schloss Dagstuhl--Leibniz-Zentrum fuer Informatik},
address = {Dagstuhl, DEU},
url = {https://doi.org/10.4230/LIPIcs.CCC.2019.8},
doi = {10.4230/LIPIcs.CCC.2019.8},
booktitle = {Proceedings of the 34th Computational Complexity Conference},
articleno = {8},
numpages = {30},
location = {New Brunswick, New Jersey},
series = {CCC '19}
}

@inproceedings{dOrsiT23,
  author       = {Tommaso {d'Orsi} and
                  Luca Trevisan},
  title        = {A Ihara-Bass Formula for Non-Boolean Matrices and Strong Refutations
                  of Random CSPs},
  booktitle    = {38th Computational Complexity Conference, {CCC} 2023, July 17-20,
                  2023, Warwick, {UK}},
  series       = {LIPIcs},
  volume       = {264},
  pages        = {27:1--27:16},
  publisher    = {Schloss Dagstuhl - Leibniz-Zentrum f{\"{u}}r Informatik},
  year         = {2023},
}

@inproceedings{BarakM16,
  author    = {Boaz Barak and
               Ankur Moitra},
  title     = {{Noisy Tensor Completion via the Sum-of-Squares Hierarchy}},
  booktitle = {Proceedings of the 29th Conference on Learning Theory, {COLT} 2016,
               New York, USA, June 23-26, 2016},
  series    = {{JMLR} Workshop and Conference Proceedings},
  volume    = {49},
  pages     = {417--445},
  publisher = {JMLR.org},
  year      = {2016},
}

@article{Coja-OghlanGL07,
  title={{Strong refutation heuristics for random $k$-SAT}},
  author={Coja{-}Oghlan, Amin and Goerdt, Andreas and Lanka, Andr{\'e}},
  journal={Combinatorics, Probability \& Computing},
  volume={16},
  number={1},
  pages={5},
  year={2007},
  publisher={Cambridge University Press}
}

@article{FlemingKP19,
  title={{Semialgebraic Proofs and Efficient Algorithm Design}},
  author={Fleming, Noah and Kothari, Pravesh and Pitassi, Toniann},
  journal={Foundations and Trends{\textregistered} in Theoretical Computer Science},
  volume={14},
  number={1-2},
  pages={1--221},
  year={2019},
  publisher={Now Publishers, Inc.}
}

@inproceedings{GuruswamiKM22,
  author    = {Venkatesan Guruswami and
               Pravesh K. Kothari and
               Peter Manohar},
  title     = {{Algorithms and certificates for Boolean {CSP} refutation: smoothed
               is no harder than random}},
  booktitle = {{STOC} '22: 54th Annual {ACM} {SIGACT} Symposium on Theory of Computing,
               Rome, Italy, June 20 - 24, 2022},
  pages     = {678--689},
  publisher = {{ACM}},
  year      = {2022}
}

@inproceedings{HsiehKM23,
  author    = {Jun{-}Ting Hsieh and
               Pravesh K. Kothari and
               Sidhanth Mohanty},
  title     = {{A simple and sharper proof of the hypergraph Moore bound}},
  booktitle = {Proceedings of the 2023 {ACM-SIAM} Symposium on Discrete Algorithms,
               {SODA} 2023, Florence, Italy, January 22-25, 2023},
  pages     = {2324--2344},
  publisher = {{SIAM}},
  year      = {2023}
}

@article{GLSS15,
author = {Gavinsky, Dmitry and Lovett, Shachar and Saks, Michael and Srinivasan, Srikanth},
title = {A tail bound for read-k families of functions},
journal = {Random Structures \& Algorithms},
volume = {47},
number = {1},
pages = {99-108},
doi = {https://doi.org/10.1002/rsa.20532},
url = {https://onlinelibrary.wiley.com/doi/abs/10.1002/rsa.20532},
eprint = {https://onlinelibrary.wiley.com/doi/pdf/10.1002/rsa.20532},
year = {2015}
}

@article{GoemansWilliamson,
  title={Improved approximation algorithms for maximum cut and satisfiability problems using semidefinite programming},
  author={Goemans, Michel X and Williamson, David P},
  journal={Journal of the ACM (JACM)},
  volume={42},
  number={6},
  pages={1115--1145},
  year={1995},
  publisher={ACM New York, NY, USA}
}

@inproceedings{CharikarMakarychevs,
author = {Charikar, Moses and Makarychev, Konstantin and Makarychev, Yury},
title = {Integrality gaps for Sherali-Adams relaxations},
year = {2009},
isbn = {9781605585062},
publisher = {Association for Computing Machinery},
address = {New York, NY, USA},
url = {https://doi.org/10.1145/1536414.1536455},
doi = {10.1145/1536414.1536455},
booktitle = {Proceedings of the Forty-First Annual ACM Symposium on Theory of Computing},
pages = {283–292},
numpages = {10},
location = {Bethesda, MD, USA},
series = {STOC '09}
}

@INPROCEEDINGS{HopkinsST,
  author={Hopkins, Samuel B. and Schramm, Tselil and Trevisan, Luca},
  booktitle={2020 IEEE 61st Annual Symposium on Foundations of Computer Science (FOCS)}, 
  title={Subexponential LPs Approximate Max-Cut}, 
  year={2020},
  volume={},
  number={},
  pages={943-953},
  doi={10.1109/FOCS46700.2020.00092}}

@inproceedings{AOW15,
  author    = {Sarah R. Allen and
               Ryan O'Donnell and
               David Witmer},
  title     = {{How to Refute a Random CSP}},
  booktitle = {{IEEE} 56th Annual Symposium on Foundations of Computer Science, {FOCS}
               2015, Berkeley, CA, USA, 17-20 October, 2015},
  pages     = {689--708},
  publisher = {{IEEE} Computer Society},
  year      = {2015},
}

@inproceedings{AbascalGK21,
  author    = {Jackson Abascal and
               Venkatesan Guruswami and
               Pravesh K. Kothari},
  title     = {{Strongly refuting all semi-random Boolean CSPs}},
  booktitle = {Proceedings of the 2021 {ACM-SIAM} Symposium on Discrete Algorithms,
               {SODA} 2021, Virtual Conference, January 10 - 13, 2021},
  pages     = {454--472},
  publisher = {{SIAM}},
  year      = {2021},
}
\appendix

\section{Proofs of auxiliary statements}
In this section, we fill in the proofs of \Cref{lem:randomcertificateexists} and \Cref{lemma:os19corollary}. The proofs here are all reformulations of ideas in prior works \cite{AOW15,OS19} and are included here for completeness. 
\subsection{Auxiliary statements for random $k$-XOR}
We show that with high probability, the matrix $A(\psi)$ defined in \Cref{def:randcertificate} is a spectral certificate.
\randomcertificateexists*
This follows directly from the following two statements.
\begin{fact}\cite[Lemma~A.5]{AOW15}\label{fact:randomspecnorm}
    Let $\psi(x)=\sum_{U\in \cH}b_U\cdot x^U$ be a $k$-XOR instance with $n^{k/2+\delta}$ expected constraints. Then with high probability, $\Norm{A(\psi)}\leq 2^{O(k)}n^{\delta}\log^3n$. 
\end{fact}

\begin{restatable}{proposition}{degsgood}\label{prop:degsgood}
    Let $k$ be an odd integer, $0<\delta<\frac12$, and $\psi(x)=\sum_{U\in \cH} b_U\cdot x^U$ be a $k$-XOR instance with $n^{k/2+\delta}$ expected constraints. Then with probability $1-\exp(-\Omega(n^{2\delta}))$, we have $\min_{(S_1,S_2)}\Norm{A(\psi)_{(S_1,S_2)}}_1\geq\frac{n^{2\delta}}{7}$.
\end{restatable}
\begin{proof}
    Since each $U\in [n]^k$ is included in $\cH$ independently with probability $p=n^{-k/2+\delta}$, for fixed $S\in [n]^{(k-1)/2}$ and $i\in [n]$, the number of $(STi)\in \cH$ is Binomial$(n^{(k-1)/2}, n^{-k/2+\delta})$, so $\#\{T:STi\in \cH\}=1$ with probability $n^{(k-1)/2-k/2+\delta}(1-n^{-k/2+\delta})^{n^{(k-1)/2}-1}\geq n^{\delta-1/2}(1-o(1))$ and $\#\{T:STi\in \cH\}=2$ with probability ${n^{(k-1)}\choose 2}n^{-k+2\delta}(1-n^{-k/2+\delta})^{n^{(k-1)/2-2}}\geq n^{2\delta-1}(\frac12-o(1))$. We now show that the number of nonzero entries in $A(\psi)_{(S_1,S_2)}$ is at least $n^{2\delta}/7$ for all $(S_1,S_2)\in [n]^{(k-1)/2}\times [n]^{(k-1)/2}$ with probability $1-\exp(-\Omega(n^{2\delta}))$, from which the result holds from a union bound. 
    
    First, suppose that $S_1\neq S_2$. For all $i\in [n]$, call $i$ a good index if $\#\{T:S_1Ti\in \cH\}=\#\{T:S_2Ti\in \cH\}=1$ and let $I$ be the set of all good indices. Note that $|I|$ is Binomial with mean $n^{2\delta}(1-o(1))$, so $ n^{2\delta}/2\leq|I|\leq 2n^{2\delta}$ by Chernoff with probability $1-\exp(-\Omega(n^{2\delta}))$. For any $i\in I$, let $T^i_1,T^i_2$ be the unique tuples such that $S_1T^i_1i,S_2T^i_2i\in \cH$. Note that conditioned on $i\in I$, $T^i_1,T_2^i$ are uniform and independent random elements of $[n]^{(k-1)/2}$. Now, for all $i,j\in I$ such that $i<j$ and $\{T_1^i,T_2^i\}\cap\{T_1^j,T_2^j\}\neq\emptyset$, remove $j$ from $I$. Furthermore, remove $j$ from $I$ if $\{T_1^j,T_2^j\}\cap\{S_1,S_2\}\neq\emptyset$. We bound the number of elements removed. Note that $j\in I$ is removed if $T_1^j$ or $T_2^j$ appears in $\bigcup_{i\in I:i<j}\{T_1^i,T_2^i\}$, which has at most $4n^{2\delta}$ elements. The probability that we remove more than $n^{2\delta}/6$ elements of $I$ is thus ${|I|\choose n^{2\delta}/6}\left(\frac{2(4n^{2\delta}+2)}{n^{(k-1)/2}}\right)^{n^{2\delta}/6}=\exp(-\Omega(n^{2\delta}))$. Thus $|I|\geq n^{2\delta}/3$ with probability $1-\exp(-\Omega(n^{2\delta}))$.

    Let $A(\psi)_{(S_1,S_2)}\coloneqq u+v$, where $u$ is the contribution of the indicies in $I$, i.e. $u_{(T_1,T_2)}\coloneqq\sum_{i\in I}\1\{T_1=T_1^i,T_2=T_2^i\}b_{S_1T_1i}b_{S_2T_2i}$, and $v=A(\psi)_{(S_1,S_2)}-u$. Since for any good $i\in I$, $S_1Ti,S_2Ti\notin \cH$ for all $T\notin \{T_1^i,T_2^i\}$, we see that $u$ is independent of $v$ when conditioned on $\cH$. Also conditioning on $v$, for all $i\in I$, we have $u_{(T^i_1,T^i_2)}=-v_{(T^i_1,T^i_2)}$ independently of all other $j\in I$ with probability at most $\frac12$, so $A(\psi)_{(S_1,S_2),(T^i_1,T^i_2)}=0$ with probability at most $\frac12$. Thus by a Chernoff bound, $A(\psi)_{(S_1,S_2),(T^i_1,T^i_2)}\neq 0$ for at least $n^{2\delta}/7$ values of $i\in I$ with probability $1-\exp(-\Omega(n^{2\delta}))$.
    
    Now consider the case where $S_1=S_2\coloneqq S$. For all $i\in [n]$, call $i$ a good index if $\#\{T:STi\in\cH\}=2$ and let $I$ be the set of all good indicies. Then $|I|$ is Binomial with mean $n^{2\delta}(\frac12-o(1))$, so $2n^{2\delta}/5\leq|I|\leq n^{2\delta}$ with probability $1-\exp(-\Omega(n^{2\delta}))$. For any $i\in I$, let $T_1^i,T_2^i$ be the unique pair of tuples such that $ST_1^i,ST_2^i\in \cH$. Note that conditioned on $i\in I$, $\{T_1^i,T_2^i\}$ is a uniformly random pair in $[n]^{(k-1)/2}$. Now, for all $i,j\in I$ such that $i<j$ and $\{T_1^i,T_2^i\}\cap\{T_1^j,T_2^j\}\neq\emptyset$, remove $j$ from $I$, and also remove $j$ if $S\in \{T_1^j,T_2^j\}$. The probability that we remove more than $n^{2\delta}/15$ elements of $I$ is ${|I|\choose n^{2\delta}/15}\left(\frac{2(2n^{2\delta}+2)}{n^{(k-1)/2}}\right)^{n^{2\delta}/15}=\exp(-\Omega(n^{2\delta}))$, so $|I|\geq n^{2\delta}/3$ with probability $1-\exp(-\Omega(n^{2\delta}))$.
    
    Define $u$ and $v$ as before, where $u$ is the contributions of the indicies in $I$, so $u_{(T_1,T_2)}\coloneqq\sum_{i\in I}\1\{T_1=T_1^i,T_2=T_2^i\}b_{S_1T_1i}b_{S_2T_2i}+\1\{T_1=T_2^i,T_2=T_1^i\}b_{S_1T_1i}b_{S_2T_2i}$ As before, $u$ is independent of $v$ when conditioned on $\cH$, and when conditioned on $v$, for all $i\in I$, we have $u_{(T^i_1,T^i_2)}=-v_{(T^i_1,T^i_2)}$ independently of all other $j\in I$ with probability at most $\frac12$. Then $A(\psi)_{(S,S),(T_1^i,T_2^i)}=0$ with probability at most $\frac12$ and thus it is nonzero for at least $n^{2\delta}/7$ with probability $1-\exp(-\Omega(n^{2\delta}))$.
    
\end{proof}

\saspectralbd*
We utilize the following reparameterization of results from \cite{OS19}:
\begin{fact}\cite[Theorem~3.6, Corollary~3.7]{OS19}\label{fact:2xorOS19}
Let $G = (V,E)$ be a multigraph with signs $b_{uv}$ for all $uv\in E$ and signed transition operator $\ol{K}$.  Suppose we have $\pi_* = \min_{v \in V} \frac{\deg(v)}{2|E|}$.
Given $\gamma > \min(\pi_*^{-1/2},\rho(\ol{K}))$, take $\ell = \Big\lceil \frac{1}{4} \frac{\log (4\gamma^2 \pi_*/25)}{\log (5\rho(\ol{K})/2\gamma)}\Big\rceil$, and take $k = \lceil (\frac{5}{2\gamma})^{2\ell}\rceil$.
Then for $R = k\ell + 1$, it holds that $R$-local degree-$2$ SOS can deduce that for all $x\in \{\pm 1\}^V$,
\begin{align*}
    -\gamma\leq \frac{1}{|E|}\sum_{uv\in E}b_{uv}x_ux_v\leq \gamma\mper
\end{align*}
\end{fact}
\begin{proof}[Proof of \Cref{lemma:os19corollary}]
    For all $S\subseteq [n]^t$, define a variable $y_S$ and define a $R$-local pseudoexpectation $\cD'$ over the $y_S$s such that for all $S_1, \cdots, S_r$ where $r\leq R$, we have $\pE_{\cD'}y_{S_1}y_{S_2}, \cdots y_{S_r}\coloneqq \pE_\cD \prod_{i\in S_1}x_i\prod_{i\in S_2}x_i\cdots\prod_{i\in S_r}x_i$. Then applying \Cref{fact:2xorOS19} with $\eta=\gamma$, we see that $\pE_{\cD'}\sum_{S,T}y_S y_TA_{S,T}\leq \eta \sum_{S,T}|A_{S,T}|$, and therefore
\begin{align*}
    \pE_\cD (x^{\odot t})^\top Ax^{\odot t}=\sum_{S,T}\pE_\cD A_{S,T}\prod_{i\in S}x_i\prod_{i\in T}x_{i}= \pE_{\cD'}\sum_{S,T}A_{S,T} y_S y_T\leq \eta \sum_{S,T}|A_{S,T}|\mper
\end{align*}
    
\end{proof}
\subsection{Auxiliary statements for semirandom $k$-XOR}
\inducedtoorig*
\begin{proof}
    Let $\cD$ be a $kR$-local pseudodistribution over variables $x_1, \cdots, x_n$. Letting $\cS$ be the variable set of $\phi$, we define a $R$-local pseudodistribution $\cD'$ over variables $y_S$ of $\phi$ such that for all $S_1, \cdots, S_r$ where $r\leq R$, we have $\pE_{\cD'}y_{S_1}y_{S_2}, \cdots y_{S_r}\coloneqq \pE_\cD \prod_{i\in S_1}x_i\prod_{i\in S_2}x_i\cdots\prod_{i\in S_r}x_i$. Then we have
    \begin{align*}
        \pE_{\cD}\psi(x)=\pE_\cD\sum_{(S,T)\in \cB} b_{f(S,T)}\prod_{i\in S}x_i\prod_{i\in T}x_i
        =\pE_{\cD'}\sum_{(S,T)\in \cB}b_{f(S,T)}y_S y_T=\pE_{\cD'}\phi(y)\leq 2\gamma|\cH|\mper
    \end{align*}
\end{proof}
\canprune*
We use the following concentration result for dependent indicators.
\begin{fact}[Chernoff Bound for Weakly Dependent Random Variables]\cite{GLSS15} \label{fact:dependentchernoff}
Suppose we have a set $S$ of independent random variables and indicator random variables $Y_1, \cdots, Y_r$ with mean $\mu$, such that each $Y_i$ is a function of some $T_i\subseteq S$ and no $s\in S$ is contained in more than $k$ total $T_i$'s. Then for any $\eps > 0$, we have 
\begin{align*}
    \Pr[Y_1 + \cdots + Y_r \notin (\mu \pm \epsilon)r]
\leq
2e^{-2\epsilon^2 r / k}\mper
\end{align*}
\end{fact}
\begin{proof}[Proof of \Cref{lem:canprune}]
        Note that since $P$ has at most $2n$ parts, the total $\ell_1$ norm of $A'(\psio)$ is $\sum_i\sum_{j\in [p_i]}{|\chij|\choose 2}\geq 2n\cdot \frac12\left(\frac{|\cH|}{2n}\right)^2 (1-o(1))\geq \frac{1}{5}n^{k-1+2\delta}$. Also, since for all $(S_1,S_2),(T_1,T_2)$, there are at most $n^{\delta/2}$ total $i$ such that $S_1T_1i\in \cH$ or $S_2T_2i\in \cH$, there are at most $n^{\delta/2}$ pairs of clauses that contribute to each entry of $A'(\psi)$; in other words, every entry of $A'(\psio)$ is at most $n^{\delta/2}$. Consider an arbitrary entry $(S_1,S_2),(T_1,T_2)$. We note that $A(\psi)_{(S_1,S_2),(T_1,T_2)}$ is the sum of $A(\psio)_{(S_1,S_2),(T_1,T_2)}$ products of independent Rademacher random variables, all of which are distinct, and thus it is also the sum of independent Rademachers, so $A(\psi)_{(S_1,S_2),(T_1,T_2)}\leq 4k\sqrt{A(\psio)_{(S_1,S_2),(T_1,T_2)}\log n}\cdot\leq 4kn^{\delta/4}\sqrt{\log n}$ with probability $1-\frac{1}{n^{2k}}$ and thus it holds for all $n^{2k-2}$ entries with probability $\geq 1-1/n^2$.
        
        Now, while there exists some row of $A'(\psio)$ indexed by $(S_1,S_2)$ whose total sum is less than $3n^{3\delta/2}$, delete that row and its corresponding column. Then the sum of the elements of $A'(\psio)$ we delete is at most $2\cdot3n^{k-1+3\delta/2}\leq 30n^{-\delta/2}\sum_{i\in [n]}\sum_{j\in [p_i]}|\chij|(|\chij|-1)$, and first condition thus follows from the fact that the $\ell_1$ norm of rows in $A'(\psi)$ is upper bounded by the sum of the corresponding row in $A'(\psio)$. Let $A(\psi)$ and $A(\psio)$ be the matrices after the deletion step.

Now consider a surviving row indexed by some $(S_1,S_2)$, whose $\ell_1$ norm is at least $3n^{3\delta/2}$. Since every entry of $A(\psio)$ is at most $n^{\delta/2}$, there are at least $3n^\delta$ nonzero entries in the row indexed by $(S_1,S_2)$, each of which are indexed by some $(T_1,T_2)$. Recall that if an entry $A(\psio)_{(S_1,S_2),(T_1,T_2)}$ is nonzero, then $A(\psi)_{(S_1,S_2),(T_1,T_2)}$ is nonzero with probability at least $\frac12$ (\Cref{remark:allonescomparison}).
Furthermore, each entry in row $(S_1,S_2)$ of $A(\psi)$ depends on variables of the form $b_{S_1Ti}$ and $b_{S_2Ti}$. Note that anytime some $b_{S_1Ti}$ appears in row $(S_1,S_2)$, it must be multiplied with some $b_{S_2T'i}$ (in which case it corresponds to the $(T,T')$ entry of that row). However, we are given that there are at most $n^{\delta/2}$ total $S_2T'i\in \cH$, and thus $b_{S_1Ti}$ appears in at most $n^{\delta/2}$ entries in row $(S_1,S_2)$. Therefore, the family of indicators of entries in row $(S_1,S_2)$ being nonzero is read-$n^{\delta/2}$ with respect to the $b_U$s, and thus after coupling to a family of indicators which occur with probability $\frac12$, by \Cref{fact:dependentchernoff}, with probability $1-\exp(-\Omega(n^{\delta/2}))$, the number of nonzero entries in row $(S_1,S_2)$ of $A(\psi)$ is at least $\frac13$ the number of nonzero entries in row $(S_1,S_2)$ of $A(\psio)$, which is at least $n^\delta$. Union bounding over all $\leq n^{k-1}$ rows of $A'(\psi)$ gives the desired result.
\end{proof}

\end{document}